\documentclass[a4paper,american,reqno]{amsart}

\usepackage{babel}
\usepackage[utf8]{inputenc}
\usepackage[T1]{fontenc}
\usepackage{lmodern}
\usepackage{csquotes}
\usepackage{todonotes}
\usepackage{amssymb}
\usepackage{nicefrac}
\usepackage{aligned-overset}
\usepackage{xspace}
\usepackage{enumitem}
\usepackage[style = authoryear-comp,
            maxbibnames = 100,
            maxcitenames = 2,
            giveninits = true,
            uniquename = init,
            isbn = false,
            backend = bibtex]{biblatex}
\usepackage[colorlinks,
            citecolor=blue,
            urlcolor=blue,
            linkcolor=blue]{hyperref}
\usepackage{cleveref}

\makeatletter
\patchcmd{\@settitle}{\uppercasenonmath\@title}{\scshape\large}{}{}
\patchcmd{\@setauthors}{\MakeUppercase}{\scshape\normalsize}{}{}
\makeatother

\vfuzz \hfuzz
\makeatletter
\@namedef{subjclassname@2020}{%
  \textup{2020} Mathematics Subject Classification}
\makeatother

\newif\ifjournal\journalfalse
\newlist{thmparts}{enumerate}{1}
\setlist[thmparts]{
	label=\alph*)
}

\makeatletter
\apptocmd{\cref@getref}{\xdef\@lastusedlabel{#1}}{}{error}
\makeatother

\Crefformat{thmpart}{%
	\StrCount{\@lastusedlabel}{:}[\LastColonPos]%
	\StrBefore[\LastColonPos]{\@lastusedlabel}{:}[\ThmLabel]%
	\nameCref{\ThmLabel}~#2\ref*{\ThmLabel}#1#3%
}
\Crefmultiformat{thmpart}{%
	\StrCount{\@lastusedlabel}{:}[\LastColonPos]%
	\StrBefore[\LastColonPos]{\@lastusedlabel}{:}[\ThmLabel]%
	\nameCref{\ThmLabel}~#2\ref*{\ThmLabel}#1#3%
}{ and~#2#1#3}{, #2#1#3}{ and~#2#1#3}

\Crefname{assumption}{Assumption}{Assumption}
\Crefname{proposition}{Proposition}{Propositions}
\Crefname{theorem}{Theorem}{Theorems}
\Crefname{lemma}{Lemma}{Lemmas}

\newcommand{\jscom}[2][]{
	\ifthenelse{\equal{#1}{journal}}{
	\ifjournal\todo[author=JS, color=yellow!50, size=\small]{#2}\fi}{\todo[author=JS, color=yellow!50, size=\small]{#2}}
	}

\newenvironment{proofClaim}[1][]{\ifthenelse{\equal{#1}{}}{\begin{proof}}{\begin{proof}[#1]}}{\end{proof}}

\newcommand{\inte}{\mathrm{int}}
\newcommand{\con}{\mathrm{con}}

\newcommand{\BR}{\mathrm{BR}}

\newcommand{\xint}{z}
\newcommand{\xcon}{y}
\newcommand{\Vcon}{\phi}
\newcommand{\XINT}{Z}

\newcommand{\R}{\mathbb{R}}
\newcommand{\Z}{\mathbb{Z}}
\newcommand{\N}{\mathbb{N}}

\NewDocumentCommand{\freeset}{O{}}{S^{#1}}

\NewDocumentCommand{\FeasStrats}{O{}}{W_{#1}}
\NewDocumentCommand{\FeasStratsInt}{O{}}{W_{#1}^\inte}
\NewDocumentCommand{\FeasStratsCon}{O{}}{W_{#1}^\con}
\NewDocumentCommand{\FeasStratsRel}{O{}}{\hat{W}_{#1}}

\NewDocumentCommand{\acut}{O{ }}{an ANE-cut#1}
\NewDocumentCommand{\Cplus}{O{}O{ }}{$\appro_{#1}^>$-witness#2}

\NewDocumentCommand{\minapprox}{O{ }}{simple binary search method#1}

\NewDocumentCommand{\eqtuple}{O{ }}{equilibrium tuple#1}
\NewDocumentCommand{\BinBC}{O{ }}{adaptive B\&C method#1}
\NewDocumentCommand{\aBinBC}{O{ }}{an adaptive B\&C method#1}

\AddToHook{cmd/appendix/before}{\crefalias{section}{appendix}}

\newcommand{\AddThm}[2]{
	\newtheorem{#1}[thmcntr]{#2}
	\AddToHook{env/#1/begin}{\crefalias{thmcntr}{#1}}
}

\theoremstyle{definition}
\AddThm{definition}{Definition}
\AddThm{example}{Example}
\AddThm{proposition}{Proposition}
\AddThm{lemma}{Lemma}
\AddThm{theorem}{Theorem}
\AddThm{conjecture}{Conjecture}
\AddThm{question}{Question}
\AddThm{corollary}{Corollary}
\AddThm{claim}{Claim}
\AddThm{assumption}{Assumption}
\AddThm{remark}{Remark}
\AddThm{invariant}{Invariant}
\AddThm{subclaim}{Subclaim}

\newcommand{\wrt}{w.r.t.\ }

\NewDocumentCommand{\NeInt}{O{\ }}{TODO{#1}}

\newcommand{\defset}[3][\defsep]{\set{#2#1#3}}
\newcommand{\Defset}[3][\defsep]{\Set{#2#1#3}}
\newcommand{\set}[1]{\{#1\}}
\newcommand{\Set}[1]{\left\{#1\right\}}
\newcommand{\st}{\text{s.t.}}
\newcommand{\define}{\mathrel{{\mathop:}{=}}}

\bibliography{RationalEQ}

\begin{document}

\title[When do Mixed-Integer Games Admit Rational Equilibria?]%
{When do Mixed-Integer Games\\Admit Rational Equilibria?}

\author[A. Duguet, T. Harks, M. Schmidt, J. Schwarz]%
{Aloïs Duguet, Tobias Harks, Martin Schmidt, Julian Schwarz}

\address[A. Duguet, M. Schmidt]{%
  Trier University,
  Department of Mathematics,
  Universitätsring 15,
  54296 Trier,
  Germany}
\email{duguet@uni-trier.de}
\email{martin.schmidt@uni-trier.de}

\address[T. Harks, J. Schwarz]{%
  University of Passau,
  Faculty of Computer Science and Mathematics,
  94032 Passau,
  Germany}
\email{tobias.harks@uni-passau.de}
\email{julian.schwarz@uni-passau.de}

\date{\today}

\begin{abstract}
  %We consider linear-quadratic generalized Nash equilibrium problems with
%mixed-integer variables and study the question of the existence of
%\emph{rational equilibria} assuming rational input data.
We consider mixed-integer linear-quadratic generalized Nash
equilibrium problems, i.e., games in which each player
solves a mixed-integer program subject to linear constraints in her
own and rivals' strategies  as well as an objective which is
quadratic in her own strategies and bilinear
in her own  and rivals' strategies.
For this class of games, we study the question of the existence of
\emph{rational equilibria} assuming rational input data.
We distinguish four subclasses according to the presence of
player-quadratic terms in the objective and rival-dependent constraints.
As our main result, we completely settle the rationality question for
all four subclasses, i.e., we show that only player-linear games
without player-quadratic terms and without rival-dependent constraints
admit rational equilibria---if the game admits equilibria at all.
All other three classes contain instances with irrational
equilibria only.

%%% Local Variables:
%%% mode: latex
%%% TeX-master: "RationalEQ"
%%% End:
\end{abstract}

\keywords{Nash equilibrium problems,
Generalized Nash equilibrium problems,
Mixed-integer games,
Existence,
Rational vs.\ irrational equilibria%
%
%
%%% Local Variables:
%%% mode: latex
%%% TeX-master: "RationalEQ"
%%% End:
}
\subjclass[2020]{90C11, % Mixed integer programming
91Axx, % Game theory
%
%%% Local Variables:
%%% mode: latex
%%% TeX-master: "RationalEQ"
%%% End:
}

\maketitle

\section{Introduction}
\label{sec:introduction}

Consider the class of linear-quadratic
generalized Nash equilibrium problems (GNEPs)
with player set~$N = \set{1, \dotsc,n}$, where
each player~$i \in N$ solves the optimization problem
\begin{equation}
  % \label{opt: player-intro}
  \notag
  \begin{split}
    \min_{x_i\in    \R^{l_i}_+} \quad
    & \sum_{j\in N} x_j^\top Q_{ij} x_i + d_i^\top x_i \\
    \st  \quad
    & \sum_{j\in N}A_{ij} x_j \geq b_i.
  \end{split}
\end{equation}
Here, $x_i\in \R^{l_i}_+$ is the strategy of player~$i$, which we
assume (w.l.o.g.) to be non-negative.
If the matrices $Q_i$, $i\in N$ are positive semidefinite, we obtain
the classic linear-quadratic convex generalized Nash equilibrium
problem (GNEP) that contains, among others, the 2-player mixed Nash
equilibrium problem \parencite{nash1950equilibrium} as a special
case.
For such convex problems, Nash equilibria can be characterized
by the players' KKT conditions, which in turn can be stated as a
linear complementarity problem (LCP).
For rational input data, i.e., $Q_{ij}$, $A_{ij}$, $d_i$, and $b_i$
are rational for all $i,j \in N$, it follows by basic linear
programming arguments that the solution set of an LCP is either empty
or it contains a rational point.
Indeed, the pivoting algorithm by~\textcite{Lemke1965} or the
Lemke--Howson algorithm (cf.~\textcite{Lemke1964}) for bimatrix games
would output such a rational solution. 

In this paper, we consider the more general class of
\emph{mixed-integer} linear-quadratic generalized Nash equilibrium
problems of the form
\begin{equation}
  \label{opt: player}
  \tag{$\mathcal{P}_i(x_{-i})$}
  \begin{split}
    \min_{x_i} \quad
    & \sum_{j\in N} x_j^\top Q_{ij} x_i + d_i^\top x_i \\
    \st  \quad
    & \sum_{j\in N}A_{ij} x_j \geq b_i, \\
    & x_i\in \R^{l_i}_+ \times \Z^{k_i}_+,
  \end{split}
\end{equation}
where we again assume rational data $Q_{ij}$, $A_{ij}$, $d_i$, and
$b_i$ for all $i,j\in N$. 
Despite the presence of integer variables,
it is known that mixed-integer linear-quadratic \emph{optimization
  problems} with convex objective (i.e., the case $n=1$) always admits
a rational solution if the problem admits a solution at all since the
KKT conditions for fixed integer variables yield an LCP.
For the general case of $n\geq 2$, we investigate in this paper the
question whether a \emph{rational} Nash equilibrium exists in case of
the game admitting any Nash equilibrium.
Note the fundamental nature of this question
because rationality of equilibria is a \emph{necessary} prerequisite
for the existence of exact finite-time algorithms under the Turing
machine model.
Moreover, the class of mixed-integer linear-quadratic GNEPs we consider here
is highly relevant. The purely continuous sub-class has a long-standing
history in game theory (cf.~the survey of
\textcite{facchinei2010generalized}), while the purely integer sub-class
has gained significant interest over the last two decades within the
realm of integer programming
games \parencite{kirst2022branch,schwarze2023branch,dragotto2023zero,carvalho2023integer}.
Transitioning to a mixed-integer framework represents the natural next
step and very recently, a few papers started to study the
topic \parencite{sagratella2017computing,Sagratella2017AlgorithmsFG,Harks24,DSHS25,Duguet25MNE}.
In particular, no results on the rationality of mixed-integer games are, to the best of our knowledge, known.\footnote{Note that the classical three player game of \cite{Nash51} showing a unique irrational mixed-Nash equilibrium relies on the fact that the expected payoff
of the players is a multilinear \emph{cubic} function so that the
equilibrium conditions lead to a quadratic equation. Our model is linear-quadratic and does not include cubic terms.}

In the following, we distinguish four problem sub-classes
according to whether or not quadratic terms in the objective are
present and whether or not rival-dependent constraints are present:
\begin{enumerate}
\item[(i)] general player-quadratic GNEPs  (PQ-GNEP) as above;
\item[(ii)] player-linear GNEPs (PL-GNEP), where $Q_{ii}= 0$ for all
  $i\in N$;
\item[(iii)] player-quadratic NEPs (PQ-NEP), where $A_{ij}=0$ for all
  $j\neq i$ and $i\in N$;
\item[(iv)] player-linear NEPs (PL-NEP), where $Q_{ii}= 0,A_{ij}=0$
  for all $j\neq i,i\in N$.
\end{enumerate}
Our main result completely classifies all four cases.
In particular, we show that only player-linear NEPs admit rational
equilibria (if they exist) and no other class does so in general:
\begin{center}
  \begin{tabular}{ c|c|c }
    & NEP & GNEP \\
    \hline
    PL&  yes (Sec.~\ref{sec:pos-result-linear-neps}, Thm.~\ref{thm:pos}) & no (Sec.~\ref{sec:counterGNEP}, Ex.~\ref{ex:GNEP}) \\
    PQ&  no (Sec.~\ref{sec:counter-nep}) & no (Sec.~\ref{sec:counterGNEP}, Ex.~\ref{ex:GNEP})
\end{tabular}
\end{center}

%%% Local Variables:
%%% mode: latex
%%% TeX-master: "RationalEQ"
%%% End:

\section{Preliminaries}
\label{sec:problem-statement}

Let us brief\/ly introduce some additional terminology we use
throughout the paper.
We denote by $\pi_i:  \prod_{j \in N}  \R^{l_j + k_j} \to \R$
the cost function of player $i$.
Moreover, we use standard game-theoretic notation and write
$x_{-i}$ for the vector of strategies of all players except player~$i$
and call a vector of strategies $x=(x_i)_{i\in N}$ strategy profile.
For any $x_{-i}$, we define the strategy set of player $i$ by
$X_i(x_{-i})\coloneq \defset{x_i \in  \R^{l_i} \times \Z^{k_i}}{
  \sum_{j\in N}A_{ij} x_j \geq b_i }$ and call the product
$X(x)\coloneq \prod_{i \in N}X_i(x_{-i})$ the joint strategy set at $x$.
We call a strategy profile $x$ feasible if $x \in X(x)$ holds.

A  feasible strategy profile~$x^* \in X(x^*)$ is called a
(generalized) Nash equilibrium if for all $i\in N$ the following
holds:
\begin{equation*}
  \pi_i(x_i^*,x_{-i}^*) \leq \pi_i(x_i,x_{-i}^*)
  \quad  \text{ for all } x_i \in X_i(x_{-i}^*).
\end{equation*}

We use the following notation throughout the paper.
If nothing else is stated, we assume that a strategy profile is of the
form $x=(\xcon,\xint)$ where $\xcon_i \define (x_{i,1}, \dotsc, x_{i,
  l_i})$ denotes the continuous components of player $i$'s strategy
and analogously $\xint_i \define (x_{i,l_i+1}, \dotsc, x_{i,l_i+k_i})$
the integer components.
We use the analogue notation for the entire and partial strategy
profiles~$x$ and~$x_{-i}$, e.g., we abbreviate $\xint_{-i} \define
(\xint_j)_{j\neq i}$.

%%% Local Variables:
%%% mode: latex
%%% TeX-master: "RationalEQ"
%%% End:

\section{Positive Result for Player-Linear NEPs}
\label{sec:pos-result-linear-neps}

In this section, we prove---under our working assumption of rational
input data---that player-linear NEPs always admit a rational NE if
they admit an NE at all.
Our proof is constructive and---under the additional condition
of the set of feasible strategies being bounded---can be used to implement a finite time
algorithm determining whether an equilibrium exists and in case of
existence, outputs one, cf.~\Cref{rem: FiniteTime}.
To this end, we assume for the entire section that for all $i\in
N$, we have $Q_{ii} = 0$ and $A_{ij}=0$ for $j\neq i$.
In order to prove the result, we require the following definitions.
For any integer strategy profile $\xint^*$, we define the following
linear complementary problem (LCP) in the continuous variables
$\xcon,\lambda$ for fixed $\xint^*$:
\begin{equation}
  \label{eq: LCP}\tag{LCP$(\xint^*)$}
  \begin{split}
    \sum_{j\neq i} \big(Q_{ij}^\con\big)^\top (\xcon_j,\xint^*_j)  + d_i^\con-
    \big( A_{ii}^\con \big)^\top \lambda_i \geq 0
    & \quad\text{ for all } i \in N,\\
    \xcon_i^\top\Big(\sum_{j\neq i} \big(Q_{ij}^\con\big)^\top (\xcon_j,\xint^*_j)  + d_i^\con-
    \big( A_{ii}^\con \big)^\top \lambda_i\Big) = 0
    & \quad\text{ for all } i \in N,
    \\
    A_{ii}(\xcon_i,\xint^*_i) - b_i \geq 0
    &\quad\text{ for all } i \in N,
    \\
    \lambda_i^\top\Big(A_{ii}(\xcon_i,\xint^*_i) - b_i\Big) = 0
    &\quad\text{ for all } i \in N,
    \\
    \xcon_i\in \R^{l_i}_+, \ \lambda_i\in \R^{m_i}_+
    &\quad\text{ for all } i \in N,
  \end{split}
\end{equation}
where $m_i$ is the dimension of $b_i$ and $Q_{ij}^\con$ as well as
$A_{ii}^\con$ denote the sub-matrices corresponding to the
continuous part of $x_i$. Analogously $d_i^\con$ denotes the sub-vector
corresponding to the continuous part of $x_i$.
This LCP corresponds to the KKT conditions of the
continuous game induced by $\xint^*$.

In addition, we define the optimal-value function for any player $i
\in N$ for fixed integer strategy components $\xint^*$ as
\begin{align}\label{def: Vcon}
  \Vcon_i(\xcon_{-i};\xint^*)\coloneq
  \min_{\xcon_i}
  \Defset{\pi_i(\xcon_i,\xcon_{-i},\xint^*)}{A_{ii}^\con \xcon_i \geq
  b_i-A_{ii}^\inte \xint^*_i},
\end{align}
where $A_{ii}^\inte$ denotes the sub-matrix corresponding to the
integer part of $x_i$.
Based on this function, we define the system of
inequalities
\begin{align}
  \label{eq: IntOpt}\tag{$\Phi(\xint^*)$}
  \Vcon_i(\xcon_{-i};\xint^*)\leq
  \Vcon_i(\xcon_{-i};(\xint_i,\xint^*_{-i}))
  \quad
  \text{ for all } \xint_i\in \XINT_i \text{ and } i\in N
\end{align}
with $\XINT_i\coloneq \set{\xint_i \in \Z^{k_i}\mid \exists \xcon_i \in \R^{l_i}_+: A_{ii}^\con \xcon_i \geq
  b_i-A_{ii}^\inte \xint_i} $.
The above system is a necessary condition for a point~$(\xcon,\xint^*)$
to be an NE: For fixed $i \in N$ and $\xint_i$, the stated inequality
ensures that there exists a continuous strategy corresponding to
$\xint_i^*$   that is optimal compared to all possible unilateral
deviations with fixed integer part $\xint_i$.
The following lemma shows that, together with \eqref{eq: LCP}, we can
state a necessary and sufficient condition for the point~$(\xcon,\xint^*)$
to be an NE. In the following, we say that $\xcon^*$ solves \eqref{eq: LCP}
if there exist corresponding $\lambda^*_i$, $i\in N$, such that $\xcon^*$
and $\lambda^*_i$, $i \in N$, solve \eqref{eq: LCP}.

\begin{lemma}\label{lem: NEChara}
  Consider an arbitrary strategy profile $x^* = (\xcon^*,\xint^*)$.
  Then, $x^*$ is an NE if and only if  $\xcon^*$ solves \eqref{eq: LCP}
  and \eqref{eq: IntOpt} simultaneously.
\end{lemma}
\begin{proof}
  We start by making two observations.
  \begin{description}[leftmargin=5pt]
  \item[1. Observation]
    Note that $\xcon^*$ solves~\eqref{eq: LCP} if and only if
    $\xcon_i^*$ is an optimal solution to
    \begin{align}
      \label{eq: PlayerCon}
      \min_{\xcon_i} \quad \pi_i(\xcon_i,\xcon_{-i},\xint^*)
      \quad \st \quad
      A_{ii}^\con\xcon_i \geq b_i -A_{ii}^\inte\xint^*_i
    \end{align}
    for each $i\in N$.
    This is true as the LCP in \eqref{eq: LCP} for a fixed $i\in N$
    contains the KKT conditions for the above problem.
    These KKT conditions are necessary and sufficient
    optimality conditions as the above problem is an LP.
    Moreover, we note that in case that $\xcon^*$ solves \eqref{eq:
      LCP}, then
    \begin{equation}
      \label{eq: ConOpt}
      \pi_i(x^*) = \Vcon_i(\xcon^*_{-i};\xint^*)
      \quad \text{for all } i\in N.
    \end{equation}
  \item[2. Observation]
    In the situation of \eqref{eq: ConOpt},
    $\xcon^*$ solves \eqref{eq: IntOpt} if and only if for all $i\in
    N$, we have
    \begin{equation*}
      \pi_i(x^*)\leq \pi_i(\xcon_i,\xint_i,\xcon_{-i}^*,\xint_{-i}^*)
      \quad \text{for all } (\xcon_i,\xint_i) \text{ with }
      A_{ii}(\xcon_i,\xint_i)\geq b_i,
    \end{equation*}
    which is exactly the equilibrium condition for $x^*$.
  \end{description}

  We now prove both directions of the lemma separately.
  \begin{description}[leftmargin=5pt]
  \item[``$\Rightarrow$'']
    Let $x^*$ be an NE. Then, every player is optimal \wrt unilateral
    deviations in her continuous variables, i.e., $\xcon_i^*$ solves
    \eqref{eq: PlayerCon} for each player $i\in N$. Hence, by our
    first observation, $\xcon^*$ solves \eqref{eq: LCP}. The second
    observation then shows that~$\xcon^*$ solves~\eqref{eq: IntOpt} since
    $x^*$ is an NE.
  \item[``$\Leftarrow$'']
    Let $x^*$ solve \eqref{eq: LCP} and \eqref{eq: IntOpt}.
    The first observation implies~\eqref{eq: ConOpt}.
    The second observation then shows that $x^*$ is an NE by $\xcon^*$
    solving \eqref{eq: IntOpt}. \qedhere
  \end{description}
\end{proof}

In order to prove the main result of this section, we will make
use of the following lemma.
\begin{lemma}\label{lem: RationalSol}
  Consider an arbitrary strategy profile $x^* = (\xcon^*,\xint^*)$. If
  the system of (in-)equalities given by~\eqref{eq: LCP}
  and~\eqref{eq: IntOpt} admits a solution, it admits a rational one
  as well.
\end{lemma}
\begin{proof}
  We first show two claims stating that \eqref{eq: LCP} and \eqref{eq:
    IntOpt} both have solution sets that are union of rational
  polyhedra, i.e., polyhedra with rational vertices.
  \begin{claim}
    The set of solutions of \eqref{eq: LCP} is a union of
    rational polyhedra.
  \end{claim}
  \begin{proofClaim}
    To see this, note that the set of solutions of \eqref{eq: LCP} is
    the union over all possible supports~$J^i\subseteq
    \set{1,\ldots,m_i}$, $L_i\subseteq \{1,\ldots,l_i\}$, $i\in N$ of solutions to
    \begin{align*}
      \sum_{j\neq i} \big(Q_{ij}^\con\big)^\top (\xcon_j,\xint^*_j)  + d_i^\con-
      \big( A_{ii}^\con \big)^\top \lambda_i \geq 0
      & \quad\text{ for all } i \in N,\\
      \Big(\sum_{j\neq i} \big(Q_{ij}^\con\big)^\top (\xcon_j,\xint^*_j)  + d_i^\con-
      \big( A_{ii}^\con \big)^\top \lambda_i\Big)_l = 0
      & \quad\text{ for all }l \in L^i \text{ and } i \in N,
      \\
      \xcon_{il} = 0
      &\quad\text{ for all } l\notin L^i \text{ and }i \in N,
      \\
      A_{ii}(\xcon_i,\xint^*_i) - b_i \geq 0
      &\quad\text{ for all } i \in N,
      \\
      \Big(A_{ii}(\xcon_i,\xint^*_i) - b_i\Big)_j = 0
      &\quad\text{ for all } k \in J^i \text{ and } i \in N,
      \\
      \lambda_{ij} = 0
      &\quad\text{ for all } j\notin J^i \text{ and }i \in N,
      \\
      \xcon_i\in \R^{l_i}_+, \ \lambda_i\in \R^{m_i}_+
      &\quad\text{ for all } i \in N,
    \end{align*}
    see Page~144 of \textcite{Cottle_et_al:2009} as well.
    The above system is linear in $\xcon$ and $\lambda$ and hence
    describes a polyhedron. Since all appearing parameters are
    rational, it follows that any vertex of the polyhedron is rational
    as well.
  \end{proofClaim}

  \begin{claim}\label{claim: SSetIntOpt}
    The set of solutions of \eqref{eq: IntOpt} is a union of
    rational polyhedra.
  \end{claim}
  \begin{proofClaim}
    We start by noting that \eqref{eq: IntOpt} is a system of
    inequalities in which both sides are expressed via a piecewise
    affine function, where every piece is defined via rational
    parameters. This is true as $\phi_i(\xcon_{-i};\xint^*)$ is the
    optimal-value function of a rational LP in which the objective
    function is linearly perturbed via $\xcon_{-i}$;
    see Lemma~\ref{lemma:opt-val-func-of-lp} in the appendix.
    Such a system is known to admit a solution set describable as
    union of rational polyhedra; cf.~\Cref{lem: SystemPiecewise} in
    the appendix.
  \end{proofClaim}

  Hence, the solution set of the system given by~\eqref{eq: LCP}
  and~\eqref{eq: IntOpt} is the intersection of unions of rational
  polyhedra, i.e.,
  \begin{equation*}
    \left(\bigcup_{l=1}^{L_1} \mathrm{P}_l^{\mathrm{LCP}}\right) \cap
    \left(\bigcup_{l=1}^{L_2} \mathrm{P}_l^{\Phi}\right) =
    \bigcup_{l_1\in L_1,l_2\in L_2} \left( \mathrm{P}_{l_1}^{\mathrm{LCP}}
      \cap  \mathrm{P}_{l_2}^{\Phi} \right).
  \end{equation*}
  Using the second representation of the solution set, we can deduce
  that if there exists a solution to \eqref{eq: LCP} and \eqref{eq: IntOpt},
  then there need to exist at least one $l_1 \in L_1$ and $l_2\in L_2$
  with $\mathrm{P}_{l_1}^{\mathrm{LCP}} \cap  \mathrm{P}_{l_2}^{\Phi}$
  not being empty. Since both of these polyhedra are rational, the
  intersection is again a rational polyhedron.
  Moreover, we just argued that this polyhedron is non-empty.
  Hence, it admits a rational vertex, which solves~\eqref{eq: LCP}
  and~\eqref{eq: IntOpt}.
\end{proof}

Using the above results, we can now easily state and prove the main
result of this section.
\begin{theorem}\label{thm:pos}
  If $G$ admits a Nash equilibrium, it also admits a rational Nash
  equilibrium.
\end{theorem}
\begin{proof}
  Let $x^*=(\xcon^*,\xint^*)$ be an NE of $G$.
  By \Cref{lem: NEChara}, $\xcon^*$ solves
  \eqref{eq: LCP} and \eqref{eq: IntOpt}.
  Hence, \Cref{lem: RationalSol} implies that there exist a rational
  $\xcon'$ solving \eqref{eq: LCP} and \eqref{eq:
    IntOpt}. Applying \Cref{lem: NEChara} shows that
  $(\xcon',\xint^*)$ is an NE as well.
\end{proof}

\begin{remark}[Finite-Time Algorithm]\label{rem: FiniteTime}
  Assume that there exists only finitely many different
  feasible integer components $\xint$, i.e., integer components
  $\xint$ for which a corresponding feasible strategy profile
  $x =(\xcon,\xint)$ exists. Then, our proof of the above theorem
  directly leads to a finite-time algorithm that can decide
  whether the game admits an equilibrium and if so, outputs a (rational) one.
  The algorithm iterates over all possible feasible integer components $\xint^*$
  and either finds a rational vector solving the corresponding
  system composed of \eqref{eq: LCP} and \eqref{eq: IntOpt} or determines
  that the system does not admit any solution.
  Note that this is possible in finite time as the solution set to
  this system is the union of solution sets of LCPs, which in turn can
  be solved in finite time.
  For the former claim, note that by \Cref{claim: SSetIntOpt}, there
  exist rational polyhedra $P_l$, $l=1,\ldots,L$, describing the
  solution set to \eqref{eq: IntOpt} and that  these polyhedra can be
  determined in finite time; see \Cref{lem: 1Piecewise,lem:
    SystemPiecewise,lemma:opt-val-func-of-lp} for their construction.
  Hence, solving for every $l=1,\ldots,L$ the LCP composed of
  \eqref{eq: LCP} augmented with the condition that $\xcon \in P_l$
  yields the entire solution set to the system composed of \eqref{eq:
    LCP} and \eqref{eq: IntOpt}.
\end{remark}

%%% Local Variables:
%%% mode: latex
%%% TeX-master: "RationalEQ"
%%% End:

\section{Negative Result for Player-Quadratic NEPs}\label{sec:counter-nep}

We now give an example for a player-quadratic NEP with 3~players that
admits a unique Nash equilibrium having irrational continuous strategies.
Note that in this as well as the following section, we do not represent
the problem of a player in the form of \eqref{opt: player} but also
allow for negative variables.
Yet, it can be easily shown that they can be brought into the form of
\eqref{opt: player} while keeping the property of admitting a single
equilibrium having irrational continuous strategies.

The first player solves
\begin{equation}
  \tag{$\mathcal{P}_1(x_2,x_3)$}
  \label{model:NEP1}
  \begin{split}
    \min_{\xcon_1,\xint_1}\quad
    &\pi_1(\xcon_1,\xint_1;\xcon_2,\xint_2,\xcon_3,\xint_3) \coloneq
    \xcon_1^2 - \xcon_3\xcon_1 + \xcon_3\xint_1  -
    (8\xint_3-1)\xint_1
    \\
    \st \quad
    &-5\xint_1 \leq \xcon_1 \leq 5\xint_1, \\
    &\xcon_1 \in \mathbb R, \ \xint_1 \in \{0,1\}.
  \end{split}
\end{equation}
The second player solves
\begin{equation}
  \tag{$\mathcal{P}_2(x_1,x_3)$}
  \label{model:NEP1}
  \begin{split}
    \min_{\xcon_2,\xint_2}\quad
    &\pi_2(\xcon_2,\xint_2;\xcon_1,\xint_1,\xcon_3,\xint_3) \coloneq
    \xcon_2^2 - \xcon_3\xcon_2 + \xcon_3\xint_2 +(30\xint_3
    +1)\xint_2
    \\
    \st  \quad
    &-5\xint_2 \leq \xcon_2 \leq 5\xint_2, \\
    &\xcon_2 \in \mathbb R, \ \xint_2 \in \{0,1\}.
  \end{split}
\end{equation}
Finally, the third player solves
\begin{equation}
  \tag{$\mathcal{P}_3(x_1,x_2)$}
  \label{model:P3}
  \begin{split}
    \min_{\xcon_3,\xint_3}\quad
    &\pi_3(\xcon_3,\xint_3;\xcon_1,\xint_1,\xcon_2,\xint_2)\coloneq
    \left(\xint_1-\xint_2-\frac{1}{2}\right)\xint_3
    \\
    \st  \quad
    &\xcon_3 \in [-5,4], \ \xint_3 \in \{0,1\}.
  \end{split}
\end{equation}
We first make the following three observations.
\begin{claim}\label{claim: xint_3^*=0}
  If $x^*$ is an equilibrium, then $\xint_3^*=0$ holds.
\end{claim}
\begin{proofClaim}
  Assume for the sake of a contradiction that $\xint_3^*=1$.
  Then, optimality of player~3 implies
  $\xint_1^*-\xint_2^*\leq\nicefrac{1}{2}$ and, hence, either
  $\xint_1^* = \xint_2^* = 1$ or $\xint^*_1 = 0$.
  Let us consider the first case.
  If $\xint_1^* = \xint_2^* = 1$, then player $2$ does not play
  optimally.
  Indeed, her costs for $\xint_2=0=\xcon_2$ are $0$ whereas for
  $\xint_2^* = 1$, we can bound her costs from below by
  \begin{align*}
    \pi_2(\xcon_2,\xint_2^*;\xcon_1,\xint_1^*,\xcon_3,\xint_3^*)
    =\xcon_2^2 - \xcon_3\xcon_2 + \xcon_3  +(30 +1) \geq
    0-5\cdot5-5+31 \geq 1>0.
  \end{align*}
  for any $\xcon_2$, where we used $\xcon_2\in[-5,5]$ and
  $\xcon_3\in[-5,4]$.

  Now consider the case $\xint^*_1 = 0$.
  Player 1's constraint then implies $\xcon_1^*=0$ and, hence,
  $\pi_1(x^*) = 0$.
  This is not optimal for player 1 as for $\xcon_1 = 0$ and
  $\xint_1=1$, we have $\pi_1(x_1,x^*_{-1}) = \xcon_3 -(8-1) \leq
  -3<0$  by $\xcon_3\leq 4$.
\end{proofClaim}
\begin{claim}
  \label{claim: x_2^*=0}
  If $x^*$ is an equilibrium, then $x^*_2 =(\xcon_2^*,\xint_2^*)=0$ holds.
\end{claim}
\begin{proofClaim}
  Assume for the sake of a contraction that $x^*_2\neq 0$. By the
  constraint of player 2, this implies that $\xint_2^*=1$.
  Considering the factor multiplied with $\xint_3$ in the objective of
  player 3, $(\xint_1^*-\xint_2^*-0.5)\leq -0.5$
  gives that $\xint_1^* \leq 1$ and hence the optimality of player 3 implies
  $\xint_3^* = 1$ contradicting \Cref{claim: xint_3^*=0}.
\end{proofClaim}

\begin{claim}
  \label{claim: x^*_1neq0}
  If $x^*$ is an equilibrium, then $x^*_1=(\xcon_1^*,\xint_1^*)\neq 0$ holds.
\end{claim}
\begin{proofClaim}
  Assume for the sake of a contraction that $x_1^* =0$ holds. By
  \Cref{claim: x_2^*=0}, we have $x_2^*=0$ as well. But this implies
  by optimality of player 3 that $\xint_3^* = 1$ in contradiction to
  \Cref{claim: xint_3^*=0}.
\end{proofClaim}

We show in the following that only one NE exists, which has an
irrational component.
To this end, we compute the best-response maps $\BR_i(x_{-i}^*)$
and show that the fixed point condition $x^*\in
\BR(x^*)\coloneq\times_{i\in N}\BR_i(x^*_{-i})$ is fulfilled by a
single $x^*$ admitting at least one irrational component.
By the above claims, we already know that $\xint_3^* = 0$ holds in an
NE.
Moreover, $\xcon_3$ does not have an impact on the costs of player
$3$ and hence  $x^*$ fulfills the fixed-point condition
if and only if $x^*_i \in \BR_i(x^*_{-i})$ for $i=1,2$ and $\xint_3^* =0$.
Hence, it is enough to calculate the best response maps
$\BR_1(x_2,\xcon_3,0)$ and $\BR_2(x_1,\xcon_3,0)$.

\begin{claim}
  \label{claim: BR}
  The best response map of player $1$ and $2$ fulfill
  $\BR_1(x_2,\xcon_3,0) = \BR_2(x_1,\xcon_3,0) = \BR(\xcon_3,0)$ for
  all $x_1,x_2,\xcon_3$ with
  \begin{align*}
    \BR(\xcon_3,0)=
    \begin{cases}
      (\nicefrac{\xcon_3}{2},1),
      &\text{ if } \xcon_3 \in\R\setminus
      \left[2-2\sqrt{2},2+2\sqrt{2}\right]
      \\
      \{(\nicefrac{\xcon_3}{2},1)\}\cup\{(0,0)\},
      &\text{ if } \xcon_3 \in \{ 2-2\sqrt{2},2+2\sqrt{2} \}
      \\
      (0,0),
      &\text{ if } \xcon_3 \in
      \left]2-2\sqrt{2},2+2\sqrt{2}\right[.
    \end{cases}
  \end{align*}
\end{claim}
\begin{proofClaim}
  Observe that it is enough to show that the best response mapping
  of player~1 is of the claimed form as player~2 solves the exact
  same problem as player 1 for $\xint_3 =0$. Let us consider player
  1 in the following.
  For $\xint_1$, player $1$ has two options: If he plays
  $\xint_1=0$, then $\xcon_1=0$ and he gets costs of $0$.
  Let us calculate in the following the optimal $\xcon_1$ for
  $\xint_1 = 1$ and the corresponding objective value.
  For $\xint_1=1$, the optimal $\xcon_1$ solves
  \begin{align*}
    \min_{\xcon_1}\quad
    &\pi_1(\xcon_1,1;\xcon_2,\xint_2,\xcon_3,0) = \xcon_1^2 -
      \xcon_3\xcon_1 + \xcon_3 + 1
    \\
    \st  \quad
    &-5 \leq \xcon_1 \leq 5, \\
    &\xcon_1 \in \mathbb R.
  \end{align*}
  The objective is an upward opening parabola and its minimal point
  is at $\xcon_1 = \nicefrac{\xcon_3}{2}$ which is contained in
  $[-5,5]$ by $\xcon_3 \in [-5,4]$. The corresponding objective
  value is given by
  \begin{align*}
    \pi_1(\nicefrac{\xcon_3}{2},1,x_2,\xcon_3,0)
    = \frac{1}{4}\xcon_3^2 - \frac{1}{2}\xcon_3^2 + \xcon_3 +1
    = -\frac{1}{4} \xcon_3^2 +\xcon_3 +1 .
  \end{align*}
  Hence, player 1 plays $\xint_1=1$ and $\xcon_1 =
  \nicefrac{\xcon_3}{2}$ in a best-response if and only if the above
  value is smaller or equal to $0$, which are the costs when playing
  $\xint_1=0$.
  Now $\pi_1(\nicefrac{\xcon_3}{2},1,x_2,\xcon_3,0)$ is a downwards
  opening parabola and thus is exactly 0 at the roots and below zero
  in case that $\xcon_3$ is outside of the interval between the roots
  of the parabola.
  The roots are
  \begin{align*}
    \frac{-1\pm
    \sqrt{1-4\cdot\frac{-1}{4}\cdot1}}{2\cdot\frac{-1}{4}}
    = 2 \pm 2\sqrt{2}
  \end{align*}
  which results in the claimed best-response sets.
\end{proofClaim}

We are now in the position to show that there is exactly one
irrational equilibrium.
Assume that $x^*$ is an equilibrium.
We know $x^*_2=0$ by \Cref{claim: x_2^*=0} and \Cref{claim: BR} for
player $2$ implies that $\xcon_3^*\in [2-2\sqrt{2},2+2\sqrt{2}]$ has
to hold.
Using \Cref{claim: x^*_1neq0} and \Cref{claim: BR} for player 1, we
further get $\xcon_3^*\in \R\setminus]2-2\sqrt{2},2+2\sqrt{2}[$.
Thus, we have $\xcon_3^*\in \{2-2\sqrt{2},2+2\sqrt{2}\} \cap
[-5,4] = \{2-2\sqrt{2}\}$ and by \Cref{claim: x^*_1neq0} and
\Cref{claim: BR} for player~1, we further get
$x^*_1=(\nicefrac{\xcon^*_3}{2},1)$.
Hence, we have shown that there is a single candidate for an
equilibrium~$x^*$ with an irrational component.
Moreover, it is now an immediate consequence of
\Cref{claim: BR} (and the argumentation preceding \Cref{claim: BR})
that~$x^*$ is indeed an NE.

%%% Local Variables:
%%% mode: latex
%%% TeX-master: "RationalEQ"
%%% End:
\section{Negative Result for Player-Quadratic and Player-Linear GNEPs}\label{sec:counterGNEP}

The following example describes a rational mixed-integer GNEP with
player-linear costs that admits a unique GNE having irrational
components.
As the class of rational mixed-integer GNEPs with player-linear costs
is a sub-class of rational mixed-integer GNEPs with player-quadratic
costs, the example also proves that the latter class does not
necessarily admit a rational equilibrium if an equilibrium exists.

\begin{example}\label{ex:GNEP}
  We study the Nash equilibria of the following 3-player game.
  Player~1 solves the problem
  \begin{align}
      \min_{\xcon_1,\xint_1}\quad
      &\pi_1(\xcon_1,\xint_1;\xcon_2,\xint_2,\xint_3)=
      -(2\xcon_2+4\xint_2)\xcon_1
      +\Big(1-3\xcon_2+\frac{14}{5}\xint_2-2\xint_3\Big) \xint_1
      \nonumber \\
      \st  \quad
      &0 \leq \xcon_1 \leq \xint_1, \nonumber \\
      &\xcon_1+\xcon_2 \leq 1,
        \tag{$\mathcal{P}_1(x_2,\xint_3)$}\label{model:P1}\\
      &\xcon_1 \in \mathbb R, \xint_1 \in \{0,1\},\nonumber
  \end{align}
  player 2 solves
  \begin{align}
      \min_{\xcon_2,\xint_2}\quad
      &\pi_2(\xcon_2,\xint_2;\xcon_1,\xint_1,\xint_3)=
      (2\xcon_1-4\xint_1)\xcon_2 +
      \Big(1-3\xcon_1+\frac{11}{5}\xint_1-2\xint_3\Big)\xint_2
      \nonumber\\
      \st  \quad
      &0 \leq \xcon_2 \leq \xint_2, \nonumber\\
      &\xcon_1+\xcon_2 \leq 1,
      \tag{$\mathcal{P}_2(x_1,\xint_3)$}\label{model:P2}\\
      &\xcon_2 \in \mathbb R, \xint_2 \in \{0,1\},\nonumber
  \end{align}
  and player 3 faces the problem
  \begin{equation}
    \tag{$\mathcal{P}_3(x_1,x_2)$}
    \label{model:P3}
    \begin{split}
      \min_{\xint_3}\quad
      &\pi_3(\xint_3;\xint_1,\xint_2)=\left(\xint_1+\xint_2-\frac{1}{2}\right)\xint_3
      \\
      \st  \quad
      &\xint_3 \in \{0,1\}.
    \end{split}
  \end{equation}
  Using the following two claims we show that in an
  equilibrium, $\xint_3^*=0$ and $\xint_1^*=\xint_2^* =1$ holds.
  \begin{claim} \label{claim:x3_not_one}
    If $x^*$ is an equilibrium, then $\xint_3^* = 0$.
  \end{claim}
  \begin{proofClaim}
    Assume for the sake of a contradiction that $\xint_3^*=1$ holds.
    By the optimality of player~$3$, this implies that
    $\xint_1^*+\xint_2^* \leq 1/2$. This in turn implies that
    $\xint_1^*=\xint_2^*=0$ and by the constraints of player $1$ and
    $2$, we also have $\xcon_1^* = \xcon_2^* = 0$.
    Hence, we get $\pi_1(x^*) =0$.
    This, however, contradicts the optimality condition for player $1$
    as $(\xcon_1,\xint_1)=(0,1)$ yields costs of $\pi_1(0,1;x^*_{-1}) =
    -1<0$, where we used that $x^*_2 = 0$ and $\xint_3^*=1$.
  \end{proofClaim}

  \begin{claim}
    \label{claim:zwnot00}
    If $x^*$ is an NE, then $\xint_1^*=1$ and $\xint_2^*=1$.
  \end{claim}
  \begin{proofClaim}
    We start by arguing that $\xint_2^*=1$ holds.
    Assume for the sake of a contradiction that $\xint_2^*=0$
    holds. By the constraint of player~$2$, we have $\xcon_2^*=0$.
    By \Cref{claim:x3_not_one}, we also have $\xint_3^*=0$.
    Hence, the objective of player~$1$ in this case is
    $\pi_1(x^*) = \xint_1^*$.
    By optimality of player~$1$, we thus get that $x_1^*=0$.
    However, this now contradicts the optimality of player~$3$ as
    $\pi_3^*(\xint_3;0)=-\nicefrac{\xint_3}{2}$ and hence $\xint_3^* = 0$
    is not optimal.

    The proof for $\xint_1^*=1$ works analogously.
    Assume for the sake of a contradiction that $\xint_1^*=0$ holds.
    By the constraint of player~$1$, we have $\xcon_1^*=0$.
    By \Cref{claim:x3_not_one}, we also have $\xint_3^*=0$.
    Hence, the objective of player $2$ in this case is
    $\pi_2(x^*) = \xint_2^*$.
    By optimality of player~$1$, we hence get that $x_2^*=0$.
    However, this now contradicts the optimality of player~$3$ as
    before.
  \end{proofClaim}

  By the above claims, we know that we have in an equilibrium $x^*$
  that $\xint^*_3 =0$ and $\xint_1^*=\xint_2^* = 1$.
  For this case, we now describe the best-response mappings of
  player~$1$ and~$2$.
  \begin{claim}
    For $\xint_2^*=1$ and $\xint^*_3=0$, the best-response mapping of
    player $1$ for any $\xcon_2\in[0,1]$ is given by
    \begin{equation*}
      \BR_1(\xcon_2,1,0)=
      \begin{cases}
        \set{(1-\xcon_2,1)},
        &\text{if } \xcon_2 \in
        \left[0,\frac{1+\sqrt{13/5}}{4}\right[,
        \\
        \set{(1-\xcon_2,1)} \cup \set{(0,0)},
        &\text{if } \xcon_2 = \frac{1+\sqrt{13/5}}{4},
        \\
        \set{(0,0)},
        &\text{if } \xcon_2 \in \left]
          \frac{1+\sqrt{13/5}}{4},1\right].
      \end{cases}
    \end{equation*}
  \end{claim}
  \begin{proofClaim}
    For $\xint_1$, player $1$ has two options: If she plays
    $\xint_1=0$, then $\xcon_1=0$ and she gets a cost of $0$.
    Let us calculate in the following the optimal $\xcon_1$ for
    $\xint_1 = 1$ and the corresponding objective value.
    For $\xint_1=1$, the optimal $\xcon_1$ solves
    \begin{align*}
      \min_{\xcon_1}\quad
      &\pi_1(\xcon_1,1;\xcon_2,1,0) = -(2\xcon_2+4)\xcon_1
        -3\xcon_2+\frac{19}{5}
      \\
      \st  \quad
      &0 \leq \xcon_1 \leq 1, \quad
        \xcon_1 \leq 1-\xcon_2, \quad
        \xcon_1 \in \mathbb R.
    \end{align*}
    Since $\xcon_2 \geq 0$, the term $-(2\xcon_2+4)$ is negative and,
    thus, the optimal $\xcon_1$ is as large as possible, i.e.,
    $\xcon_1=1-\xcon_2$.
    The corresponding objective value is given by
    \begin{align*}
      \pi_1(1-\xcon_2,1;\xcon_2,1,0) =  2\xcon_2^2-\xcon_2-\frac{1}{5}.
    \end{align*}
    Hence, $\xint_1=1$ and $\xcon_1 = 1-\xcon_2$ is a best-response if
    and only if the above value is smaller or equal to $0$, which are
    the costs when playing $\xint_1=0$.
    The above value is an upwards opening parabola in $\xcon_2$ with
    roots
    \begin{equation*}
      \frac{1-\sqrt{13/5}}{4}<0
      \quad \text{and} \quad
      \frac{1+\sqrt{13/5}}{4}\in[0,1].
    \end{equation*}
    Subsequently, for $\xcon_2
    \in[0,1]$, the cost $\pi_1(1-\xcon_2,1;\xcon_2,1,0)$ is smaller or equal to zero if and only
    if $$\xcon_2 \in \left[0,\frac{1+\sqrt{13/5}}{4}\right],$$ leading
    to the claimed best-response mapping.
  \end{proofClaim}

  A similar proof shows the following claim.
  \begin{claim}
    For $\xint_1^*=1$ and $\xint^*_3=0$, the best-response mapping of
    player $2$ for any $\xcon_1\in[0,1]$ is given by
    \begin{equation*}
      \BR_2(\xcon_1,1,0)=
      \begin{cases}
        \set{(1-\xcon_1,1)},
        &\text{if } \xcon_1 \in\left[0,
          \frac{3-\sqrt{13/5}}{4}\right[,
        \\
        \set{(1-\xcon_1,1)} \cup \set{(0,0)},
        &\text{if } \xcon_1 = \frac{3-\sqrt{13/5}}{4},
        \\
        \set{(0,0)},
        &\text{if } \xcon_1 \in
        \left]\frac{3-\sqrt{13/5}}{4},1\right].
      \end{cases}
    \end{equation*}
  \end{claim}
  \begin{proofClaim}
    For $\xint_2$, player $2$ has two options: If she plays
    $\xint_2=0$, then $\xcon_2=0$ follows and the cost is $0$.
    Let us now calculate the optimal~$\xcon_2$ for
    $\xint_2 = 1$ and the corresponding objective value.
    For $\xint_2=1$, the optimal $\xcon_2$ solves
    \begin{align*}
      \min_{\xcon_2}\quad
      &\pi_2(\xcon_2,1;\xcon_1,1,0) = (2\xcon_1-4)\xcon_2
        -3\xcon_1+\frac{16}{5}
      \\
      \st \quad
      &0 \leq \xcon_2 \leq 1, \quad
        \xcon_2 \leq 1-\xcon_1, \quad
        \xcon_2 \in \mathbb R.
    \end{align*}
    Since $\xcon_1 \leq 1$, the term $(2\xcon_1-4)$ is negative and
    thus the optimal $\xcon_2$ is as large as possible, i.e.,
    $\xcon_2=1-\xcon_1$.
    The corresponding objective value is given by
    \begin{align*}
      \pi_2(1-\xcon_1,1;\xcon_1,1,0)
      = -2\xcon_1^2+3\xcon_1-\frac{4}{5}.
    \end{align*}
    Hence, $\xint_1=1$ and $\xcon_1 = 1-\xcon_2$ is a best-response if
    and only if the above value is smaller or equal to $0$, which are the costs
    when playing $\xint_1=0$.
    The above is a downwards opening parabola in $\xcon_1$ with roots
    \begin{equation*}
      \frac{3-\sqrt{13/5}}{4} \in [0,1]
      \quad\text{and}\quad
      \frac{3+\sqrt{13/5}}{4}>1.
    \end{equation*}
    Hence, for $\xcon_1\in[0,1]$, the cost is smaller or equal
    to zero if and only if $$\xcon_1 \in
    \left[0,\frac{3-\sqrt{13/5}}{4}\right],$$ leading to the claimed
    best-response mapping.
  \end{proofClaim}

  We can now show that there exists only one equilibrium~$x^*$ and
  that this equilibrium admits at least one irrational component.
  For what follows, assume that $x^*$ is an NE.
  Claims~\ref{claim:x3_not_one} and \ref{claim:zwnot00} imply
  $\xint_3^* = 0$ and $\xint_1^*=\xint_2^* = 1$.
  This together with the best-response maps shown in the last two
  claims leads to
  \begin{equation*}
    \xcon_1^* \in \left[0,\frac{3-\sqrt{13/5}}{4}\right],
    \quad
    \xcon_2^*\in \left[0,\frac{1+\sqrt{13/5}}{4}\right],
    \quad
    \xcon_1^* = 1-\xcon_2^*.
  \end{equation*}
  These conditions can only be satisfied for
  \begin{equation*}
    \xcon_1^* = \frac{3-\sqrt{13/5}}{4}
    \quad\text{and}\quad
    \xcon_2^* = \frac{1+\sqrt{13/5}}{4}.
  \end{equation*}
  Hence, we have argued that there can only exist a single NE
  given by
  \begin{equation*}
    x^*
    = (\xcon_1^*,\xint_1^*,\xcon_2^*,\xint_2^*,\xint_3^*)
    = \left(\frac{3-\sqrt{13/5}}{4},1,\frac{1+\sqrt{13/5}}{4},1,0\right),
  \end{equation*}
  which is irrational.
  One can easily verify that this is indeed an NE using the above
  best-response maps.
\end{example}

%%% Local Variables:
%%% mode: latex
%%% TeX-master: "RationalEQ"
%%% End:
\section{Conclusion}
\label{sec:conclusion}

Let us close this paper with a brief discussion about the
computational consequences of our results.
First, the results of Section~\ref{sec:pos-result-linear-neps}
immediately lead to a finite time algorithm for computing an
equilibrium of player-linear NEPs (that admit equilibria at all);
see Remark~\ref{rem: FiniteTime} for a more detailed discussion.
Second, the existence of instances of player-quadratic NEPs and
player-linear or player-quadratic GNEPs that only admit irrational equilibria makes it
impossible to derive finite-time algorithms in general.

This implication directly leads to the question of approximation
schemes for these classes of problems.
One could, for instance, try to design finite-time algorithms that
compute points that are $\varepsilon$-close to irrational equilibria
if only those exist.
To the best of our knowledge, methods like this have not been studied
in the past and are an interesting topic of future research.

%%% Local Variables:
%%% mode: latex
%%% TeX-master: "RationalEQ"
%%% End:

\section*{Acknowledgements}

This research has been funded by the Deutsche Forschungsgemeinschaft
(DFG) in the project 543678993 (Aggregative gemischt-ganzzahlige
Gleichgewichtsprobleme: Existenz, Approximation und Algorithmen).
We acknowledge the support of the DFG.

%%% Local Variables:
%%% mode: latex
%%% TeX-master: "RationalEQ"
%%% End:

\appendix
\section{Technical Auxiliary Results}

We prove in the following that a system of inequalities with piecewise
affine functions and rational parameters admits a solution set
describable as the union of rational polyhedra. Moreover, we show that
the optimal value function of an LP with respect to change in the objective vector
is a  piecewise affine function with rational parameters. 
Here, we use the following definition:

\begin{definition}
  \label{def: PWA}
  We say that a function~$f:\R^n\to \R$ is a \emph{piecewise affine
    function with rational parameters} if there exists a finite set of
  polyhedra $P_i \coloneq \defset{x\in \R^n}{A_ix \leq b_i}$, $i
  =1,\ldots,k$, with rational $A_i, b_i$ together with rational $c_i\in
  \mathbb{Q}^n$ and $d_i\in \mathbb{Q}$
  such that $\bigcup_{i=1}^k P_i = \R^n$
  and
  \begin{align*}
    f(x) = c_i^\top x + d_i \quad \text{if } x \in P_i.
  \end{align*}
\end{definition}

We first need the following intermediate lemma.
\begin{lemma}
  \label{lem: 1Piecewise}
  Let $f:\R^n\to \R$ be a piecewise affine function with rational
  parameters.
  Then, the set of solutions of $f(x) \leq 0$ is the union of rational
  polyhedra.
\end{lemma}
\begin{proof}
  The set of solutions of $f(x) \leq 0$ is the union over
  $i\in\{1,\ldots,k\}$ of $P_i$ intersected with $\defset{x\in
  \R^n}{ c_i^\top x + d_i\leq 0}$. The latter intersection is
  again a rational polyhedron as it is the intersection of two
  rational polyhedra.
  Hence, the lemma is proven.
\end{proof}

\begin{lemma}
  \label{lem: SystemPiecewise}
  Consider a finite set of piecewise affine functions
  $f_i$, $i=1,\ldots,k$, with rational parameters and a corresponding
  system of inequalities of the form $f_i(x)\leq 0$. Then, the
  solution set of this system is the union of rational polyhedra.
\end{lemma}
\begin{proof}
  By \Cref{lem: 1Piecewise}, we know that for any $i\in
  \{1,\ldots,k\}$, the set of solutions of $f_i(x)\leq 0$ is the union
  of rational polyhedra. Let us denote these polyhedra via $P_l^i$ with $l\in \{1,\ldots,L_i\}$
  and $L_i\in \N$.
  Hence, the set of solutions of the entire system is given by
  \begin{align*}
    \bigcap_{i=1}^k\bigcup_{l=1}^{L_i} P_l^i
    = \bigcup_{(l_i)_{i}\in \prod_{i=1}^k\{1,\ldots,L_i\}}
    \bigcap_{i=1}^k P^i_{l_i}
  \end{align*}
  and since $\bigcap_{i=1}^kP^i_{l_i}$ is a rational
  polyhedron (as it is the intersection of rational polyhedra),
  the proof is complete.
\end{proof}

We now state and proof the second promised result.
\begin{lemma}
  \label{lemma:opt-val-func-of-lp}
  Consider the optimal-value function function
  \begin{equation*}
    \varphi(c) \define \min \Defset{c^\top x}{A x \leq b}
  \end{equation*}
  of a linear optimization problem with rational constraint data~$A
  \in \mathbb{Q}^{m \times n}$ and $b \in \mathbb{Q}^{m}$.
  Then, $\varphi$ is a piecewise linear function with rational parameters.
\end{lemma}
\begin{proof}
  For any $c \in \R^n$, a solution to the linear problem
  $\defset{c^\top x}{A x \leq b}$ is attained at one of the finite
  vertices of the polyhedron given by $A x \leq b$.
  Let us denote these vertices with $x^1, \dotsc, x^s$.
  These vertices are all rational vectors as $A$ and $b$ are rational.
  Hence, the value function can be written as
  \begin{equation*}
    \varphi(c) = \min \Defset{c^\top x}{x \in \Set{x^1, \dotsc, x^s}}.
  \end{equation*}
  In particular, for any $i\in \{1,\ldots,s\}$, we have $\varphi(c) = c^\top x^i$
  if $c^\top x^i\leq c^\top x^j$ for all $j \in \{1,\ldots,s\}$. The latter condition
  is equivalent to $c\in P_i$ where $P_i$ is the rational polyhedron
  given by $P_i \coloneq \defset{c\in \R^n}{D_i c \leq 0}$ with
  \begin{equation*}
    D_i \coloneq
    \begin{bmatrix}
      (x^i)^\top-(x^1)^\top\\
      \vdots\\
      (x^i)^\top-(x^s)^\top
    \end{bmatrix}
  \end{equation*}
  Hence, $\varphi(c)$ is of the form as in \Cref{def: PWA} and, thus,
  the proof is finished.
\end{proof}

%%% Local Variables:
%%% mode: latex
%%% TeX-master: "RationalEQ"
%%% End:

\printbibliography

\enlargethispage{2cm}

\end{document}